\documentclass[3p]{elsarticle}
\usepackage{amsmath,amssymb}
\usepackage{tikz}
\usetikzlibrary{arrows,automata,positioning}

\newtheorem{lemma}{Lemma}
\newtheorem{theorem}{Theorem}
\newtheorem{proposition}{Proposition}
\newproof{proof}{Proof}

\begin{document}
\title{On state complexity of unions of binary factor-free languages}
\author[szabivan]{Szabolcs Iv\'an}
\address[szabivan]{University of Szeged, Hungary}
\begin{abstract}
In~\cite{brzozowskigalina}, it has been conjectured that if $K$ and $L$ are factor-free regular languages over a binary alphabet having state complexity $m$ and $n$, resp,
then the state complexity of $K\cup L$ is at most $mn-(m+n)+3-\min\{m,n\}$.

We disprove this conjecture by giving a lower bound of $mn-(m+n)-2-\lfloor\frac{\min\{m,n\}-2}{2}\rfloor$, which exceeds the
conjectured bound whenever $\min\{m,n\}\geq 10$.
\end{abstract}
\maketitle
\section{Introduction}
The state complexity of a regular language is the number of states of its minimal automaton.
The state complexity of a binary operation $\circ$ is a function on $n,m$: if $K$ has state complexity $n$
and $L$ has state complexity $m$, then how large can the state complexity of $K\circ L$ be?
In this note we give a lower bound for the state complexity of \emph{union} when the alphabet is \emph{binary}
and the languages $K,L$ are \emph{factor-free}.

For a survey on state complexity see~\cite{Yu00statecomplexity},
for applications of (binary) factor-free languages in cryptography and coding theory see~\cite{shyr1991free},
for results on the state complexity of operations on subclasses of regular languages see~\cite{Han06statecomplexity,Han20092537,hankaisheng,brzozowskigalina}.

\section{Notation}
We assume the reader has basic knowledge on language and automata theory.
An alphabet is a finite nonempty set $\Sigma$, a $\Sigma$-word is a finite sequence $w=a_1\ldots a_n$ of letters $a_i\in\Sigma$,
with $\varepsilon$ denoting the empty word when $n=0$, a language (over $\Sigma$) is any set of $\Sigma$-words.
The set $\Sigma^*$ of all words forms a monoid with the operation being (con)catenation, or simply
product of words given by $a_1\ldots a_n\cdot b_1\ldots b_k=a_1\ldots a_nb_1\ldots b_k$. 

A finite automaton is a system $M=(Q,\Sigma,\delta,q_0,F)$ with $Q$ being the finite set of states, $q_0\in Q$ the start state,
$F\subseteq Q$ the set of final states, $\Sigma$ is the finite nonempty input alphabet and $\delta:Q\times\Sigma\to Q$ the
transition function which is extended to $Q\times\Sigma^*\to Q$ as $\delta(q,\varepsilon)=q$ and $\delta(q,wa)=\delta(\delta(q,w),a)$.
If $M$ is understood, we simply write $qw$ for $\delta(q,w)$. The language recognized by $M$ is $L(M)=\{w\in\Sigma^*:q_0w\in F\}$.

States of the form $q_0w$ are called \emph{reachable} states of $M$.
A trap is a non-final state $p\notin F$ such that $pa=p$ for each $a\in \Sigma$ (thus, $pw=p$ for each word $w$ as well).
Two states $p,q$ are called \emph{distinguishable} if there exists a word $w$ such that exactly one of the states $pw$ and $qw$ belongs to $F$.
It is known that $M$ is minimal (that is, has the smallest possible number of states among the automata recognizing $L(M)$)
iff each pair $p\neq q$ of its states is distingushable and all its states are reachable.
Such automata are also called \emph{reduced}.
A state $q$ is empty if there is no word $w$ with $qw\in F$. In a minimal automaton, there is at most one empty state which is then a trap.
The state complexity of a regular (that is, recognizable by some finite automaton) language is the number of states
of its minimal automaton.

A word $u$ is a factor of a word $v$ if $v=xuy$ for some words $x$ and $y$. A language $L$ is factor-free if $u,xuy\in L$
imply $x=y=\varepsilon$ for any $u,x,y\in\Sigma^*$.

\section{Factor-free languages}
In~\cite{brzozowskigalina} the following lower bound was given:
\begin{proposition}
For the binary alphabet $\Sigma=\{a,b\}$ there exist regular factor-free languages $K_n$ and $L_m$ for each $n,m>6$
such that $K_n$ has state complexity $n$, $L_m$ has state complexity $m$ and $K_n\cup L_m$ has state complexity
$nm-(n+m)+3-\min\{n,m\}$.
\end{proposition}
They conjectured this bound to be tight. In the rest of this paper we give a better lower bound, at least for the case when
$m,n$ are large enough. We note that due to~\cite{brzozowskigalina}, for \emph{at least ternary} alphabets, $mn-(n+m)$ is a tight bound.

In this section we assume that $m$ and $n$ are large enough, say $\min\{m,n\}\geq 10$.

We define two reduced $n$-state automata $A_n$ and $B_n$ as follows.
In both cases, the set of states is $[n]=\{1,\ldots,n\}$, the start state is $1$, the unique accepting state is $n-1$,
state $n$ is a trap state and $p\sigma=n$ for $p\in\{n-1,n\}$ and $\sigma\in\{a,b\}$. Moreover, $1a=2$ and $1b=n$ in both cases.

For $A_n$, let $pa=p$ for each \emph{odd} $3\leq p<n-1$, $pa=n$ for each \emph{even} $2\leq p<n-1$,
let $pb=p+1$ for each $2\leq p<n-2$ and $pb=2$ for $p=n-2$.

For $B_n$, let $pa=p+1$ for each $2\leq p\leq n-1$ and $pb=p$ for each $2\leq p<n-1$.

See Figure~\ref{fig-cross} for $A_{10}$ (depicted vertically on the left) and $B_{7}$ (depicted horizontally on top).
Missing edges all go to the trap state $n$.

It is clear that all states of $A_n$ and $B_n$ are reachable.
$A_n$ is also reduced: $ab^{n-4}a$ is accepted exactly from $1$, and for each $2\leq p\leq n-2$, $b^{n-2-p}a$ is a word
accepted exactly from $p$ and not from any other state. The empty word is accepted only from $n-1$, and $n$ is a trap state,
thus each pair of states is distinguishable.

Also, $B_n$ is reduced: for each state $1\leq p\leq n-1$, $a^{n-1-p}$ is a word which is accepted exactly from $p$ and not from any other
state, and $p=n$ is a trap state, thus again, each pair of states is distinguishable.

Let $K_n$ stand for $L(A_n)$ and $L_n$ stand for $L(B_n)$. We claim that $K_n$ and $L_n$ are factor-free.
To see this, we start with a handy lemma (which is probably folklore but we were not able to find it in the literature\footnote{We would be grateful to the reviewers to suggest a source for this lemma.}).
\begin{lemma}
\label{lem-factor-free}
A reduced automaton $A=(Q,\Sigma,\delta,q_0,F)$ with $|Q|>1$ recognizes a factor-free language if and only if it satisfies all the following conditions:
\begin{itemize}
\item[i)] $q_0$ has indegree $0$;
\item[ii)] $F=\{q_f\}$ is a singleton;
\item[iii)] there is a trap state $\bot\neq q_f$ with $p\sigma=\bot$ for every $p\in\{q_f,\bot\}$, $\sigma\in\Sigma$;
\item[iv)] for any word $u\in L(A)$ and state $p\in Q-\{q_0\}$ it holds that $pu=\bot$.
\end{itemize}
\end{lemma}
\begin{proof}
Assume the conditions all hold for $A$ and let $u,v,w\in\Sigma^*$ be words with $v\in L$ and $uvw\in L$.
We show that $u=w=\varepsilon$ in this case, proving the factor-freeness of $L$.
Suppose $u\neq\varepsilon$. Then by i), $q_0u\neq q_0$. By iv), $q_0uv=\bot$ then and thus by iii), $q_0uvw=\bot\notin F$ by ii) and that $q_f\neq\bot$.
Hence $u=\varepsilon$. Supposing $w\neq\varepsilon$ we get that $q_0uvw=q_0vw=q_fw$ which is $\bot$ by iii). Thus $w=\varepsilon$.

For the other direction, let $L$ be a factor-free regular language, and let $A$ be its minimal automaton.
Suppose $A$ does not satisfy i). Then since $A$ is reduced, $q_0u=q_0$ for some nonempty word $u$.
Since $|Q|>1$ and $A$ is reduced, $F$ is nonempty, hence $q_0v\in F$ for some $v$. Thus, $uv\in L$ as well as its proper factor $v$, a contradiction.
Thus $A$ satisfies i).

Now suppose there are at least two final states $p$ and $q$ with $pu=q$ for some (nonempty) word $u$. Then since $A$ is reduced,
$p=q_0v$ for some word $v$, hence both $vu$ and its proper factor $v$ are in $L$, a contradiction. Hence, from a final state
no other final state is reachable. This means that any state reachable from a final state by a nonempty word is an empty state.
Since $A$ is reduced, this yields that there is a trap state $\bot\notin F$ and $p\sigma=\bot$ for any final state $p\in F$ and letter $\sigma\in\Sigma$.
Thus all the final states are indistinguishable, hence there is exactly one final state, thus $A$ satisfies ii) and iii).

Also, if $pu=q\neq\bot$ for some $u\in L$ and $p\in Q-\{q_0\}$, then since $A$ is reduced, $p=q_0v$ for some nonempty word $v$ and
$q$ is nonempty, say $qw\in F$ for some $w$. Then both $vuw$ and its proper factor $u$ is in $L$, a contradiction. Hence $A$ satisfies iv).
\end{proof}

Now it is easy to show that $K_n$ and $L_n$ are factor-free.

For $K_n$, we have that $A_n$ clearly satisfies Conditions i), ii) (with $q_f=n-1$) and iii) (with $\bot=n$) of Lemma~\ref{lem-factor-free}.
For iv), any word $u\in L_n$ has the form $u=awa$ with $|w|_b\,\equiv\,-1\,\mathrm{mod}\,(n-3)$. Since $ka=\bot$ for all even $2\leq k<n-2$,
$ku=\bot$ for those states. Also, $kawa=\bot$ for $k\geq n-2$ since any word of length at least two maps those states to $\bot$.
Finally, if $3\leq k<n-2$ is odd, then $kaw$ is either $\bot$ or $k-1$ which is an even number between $2$ and $n-3$ (inclusive).
For such states $k-1$ we have $(k-1)a=\bot$, thus in this case $ku=\bot$ also holds, thus Condition iv) is also verified.

For $L_n$, observe that any word belonging to $L_n$ has the form $aua$ with $|u|_a=n-4$, hence in particular any word of $L_n$ contains
exactly $n-2$ number of $a$'s and each proper factor has less number of $a$'s, showing factor-freeness of $L_n$.

\begin{figure}[scale=0.4]
\label{fig-cross}
\tikzset{initial text={}}
\begin{tikzpicture}[node distance=1.7cm,->,auto,shorten >=1pt,scale=0.5, on grid]
\node[initial,state] (L1) {$1$};
\node[state] (L2) [right of=L1] {$2$};
\node[state] (L3) [right of=L2] {$3$};
\node[state] (L4) [right of=L3] {$4$};
\node[state] (L5) [right of=L4] {$5$};
\node[state,accepting] (L6) [right of=L5] {$6$};
\node[state] (L7) [right of=L6] {$7$};

\path (L1) edge node {$a$} (L2)
      (L2) edge node {$a$} (L3)
      (L3) edge node {$a$} (L4)
      (L4) edge node {$a$} (L5)
      (L5) edge node {$a$} (L6)
      (L6) edge node {$a,b$} (L7)
      (L2) edge [loop above] node {$b$} (L2)
      (L3) edge [loop above] node {$b$} (L3)
      (L4) edge [loop above] node {$b$} (L4)
      (L5) edge [loop above] node {$b$} (L5)
      (L7) edge [loop above] node {$a,b$} (L7)
      ;
      
\node[initial,state] (KL11) [below of=L1] {$1,1$};
\node(KL12)[below of=L2]{};
\node[state](KL22)[below of=KL12]{$2,2$};
\node[state](KL32)[below of=KL22]{$3,2$};
\node[state](KL42)[below of=KL32]{$4,2$};
\node[state](KL52)[below of=KL42]{$5,2$};
\node[state](KL62)[below of=KL52]{$6,2$};
\node[state](KL72)[below of=KL62]{$7,2$};
\node[state](KL82)[below of=KL72]{$8,2$};

\node[state](KL23)[right of=KL22]{$2,3$};
\node[state](KL33)[below of=KL23]{$3,3$};
\node[state](KL43)[below of=KL33]{$4,3$};
\node[state](KL53)[below of=KL43]{$5,3$};
\node[state](KL63)[below of=KL53]{$6,3$};
\node[state](KL73)[below of=KL63]{$7,3$};
\node[state](KL83)[below of=KL73]{$8,3$};
\node[accepting,state](KL93)[below of=KL83]{$9,3$};
\node[state](KL103)[below of=KL93]{$10,3$};

\node[state](KL24)[right of=KL23]{$2,4$};
\node[state](KL34)[below of=KL24]{$3,4$};
\node[state](KL44)[below of=KL34]{$4,4$};
\node[state](KL54)[below of=KL44]{$5,4$};
\node[state](KL64)[below of=KL54]{$6,4$};
\node[state](KL74)[below of=KL64]{$7,4$};
\node[state](KL84)[below of=KL74]{$8,4$};
\node[accepting,state](KL94)[below of=KL84]{$9,4$};
\node[state](KL104)[below of=KL94]{$10,4$};

\node[state](KL25)[right of=KL24]{$2,5$};
\node[state](KL35)[below of=KL25]{$3,5$};
\node[state](KL45)[below of=KL35]{$4,5$};
\node[state](KL55)[below of=KL45]{$5,5$};
\node[state](KL65)[below of=KL55]{$6,5$};
\node[state](KL75)[below of=KL65]{$7,5$};
\node[state](KL85)[below of=KL75]{$8,5$};
\node[accepting,state](KL95)[below of=KL85]{$9,5$};
\node[state](KL105)[below of=KL95]{$10,5$};

\node(KL26)[right of=KL25]{};
\node[accepting,state](KL36)[below of=KL26]{$3,6$};
\node(KL46)[below of=KL36]{};
\node[accepting,state](KL56)[below of=KL46]{$5,6$};
\node(KL66)[below of=KL56]{};
\node[accepting,state](KL76)[below of=KL66]{$7,6$};
\node(KL86)[below of=KL76]{};
\node[accepting,state](KL96)[below of=KL86]{$9,6$};
\node[accepting,state](KL106)[below of=KL96]{$10,6$};

\node[state](KL27)[right of=KL26]{$2,7$};
\node[state](KL37)[below of=KL27]{$3,7$};
\node[state](KL47)[below of=KL37]{$4,7$};
\node[state](KL57)[below of=KL47]{$5,7$};
\node[state](KL67)[below of=KL57]{$6,7$};
\node[state](KL77)[below of=KL67]{$7,7$};
\node[state](KL87)[below of=KL77]{$8,7$};
\node[accepting,state](KL97)[below of=KL87]{$9,7$};
\node[state](KL107)[below of=KL97]{$10,7$};

\node[initial,state] (K1) [left of=KL11]{$1$};
\node[state] (K2) [below of=K1]{$2$};
\node[state] (K3) [below of=K2]{$3$};
\node[state] (K4) [below of=K3]{$4$};
\node[state] (K5) [below of=K4]{$5$};
\node[state] (K6) [below of=K5]{$6$};
\node[state] (K7) [below of=K6]{$7$};
\node[state] (K8) [below of=K7]{$8$};
\node[state,accepting] (K9) [below of=K8]{$9$};
\node[state] (K10) [below of=K9]{$10$};

\path
  (K1) edge node {$a$} (K2)
  (K2) edge node {$b$} (K3)
  (K3) edge node {$b$} (K4)
  (K4) edge node {$b$} (K5)
  (K5) edge node {$b$} (K6)
  (K6) edge node {$b$} (K7)
  (K7) edge node {$b$} (K8)
  (K8) edge node {$a$} (K9)
  (K9) edge node {$a,b$} (K10)
  (K8) edge [bend left] node {$b$} (K2)
  (K3) edge [loop left] node {$a$} (K3)
  (K5) edge [loop left] node {$a$} (K5)
  (K7) edge [loop left] node {$a$} (K7)
  ;

\path
  (KL11) edge node {$a$} (KL22)
  (KL22) edge node {$b$} (KL32)
  (KL32) edge node {$b$} (KL42)
  (KL42) edge node {$b$} (KL52)
  (KL52) edge node {$b$} (KL62)
  (KL62) edge node {$b$} (KL72)
  (KL72) edge node {$b$} (KL82)
  (KL82) edge [bend left=20] node {$b$} (KL22)
  (KL23) edge node {$b$} (KL33)
  (KL33) edge node {$b$} (KL43)
  (KL43) edge node {$b$} (KL53)
  (KL53) edge node {$b$} (KL63)
  (KL63) edge node {$b$} (KL73)
  (KL73) edge node {$b$} (KL83)
  (KL83) edge [bend left=20] node {$b$} (KL23)
  (KL24) edge node {$b$} (KL34)
  (KL34) edge node {$b$} (KL44)
  (KL44) edge node {$b$} (KL54)
  (KL54) edge node {$b$} (KL64)
  (KL64) edge node {$b$} (KL74)
  (KL74) edge node {$b$} (KL84)
  (KL84) edge [bend left=20] node {$b$} (KL24)
  (KL25) edge node {$b$} (KL35)
  (KL35) edge node {$b$} (KL45)
  (KL45) edge node {$b$} (KL55)
  (KL55) edge node {$b$} (KL65)
  (KL65) edge node {$b$} (KL75)
  (KL75) edge node {$b$} (KL85)
  (KL85) edge [bend left=20] node {$b$} (KL25)
  (KL27) edge node {$b$} (KL37)
  (KL37) edge node {$b$} (KL47)
  (KL47) edge node {$b$} (KL57)
  (KL57) edge node {$b$} (KL67)
  (KL67) edge node {$b$} (KL77)
  (KL77) edge node {$b$} (KL87)
  (KL87) edge [bend left=20] node {$b$} (KL27)
  
  (KL32) edge node {$a$} (KL33)
  (KL33) edge node {$a$} (KL34)
  (KL34) edge node {$a$} (KL35)
  (KL35) edge node {$a$} (KL36)
  (KL36) edge node {$a$} (KL37)
  (KL37) edge [loop right] node {$a$} (KL37)
  (KL52) edge node {$a$} (KL53)
  (KL53) edge node {$a$} (KL54)
  (KL54) edge node {$a$} (KL55)
  (KL55) edge node {$a$} (KL56)
  (KL56) edge node {$a$} (KL57)
  (KL57) edge [loop right] node {$a$} (KL57)
  (KL72) edge node {$a$} (KL73)
  (KL73) edge node {$a$} (KL74)
  (KL74) edge node {$a$} (KL75)
  (KL75) edge node {$a$} (KL76)
  (KL76) edge node {$a$} (KL77)
  (KL77) edge [loop right] node {$a$} (KL77)
  
  (KL82) edge node {$a$} (KL93)
  (KL83) edge node {$a$} (KL94)
  (KL84) edge node {$a$} (KL95)
  (KL85) edge node {$a$} (KL96)
  (KL87) edge node {$a$} (KL97)
  (KL93) edge node {$b$} (KL103)
  (KL94) edge node {$b$} (KL104)
  (KL95) edge node {$b$} (KL105)
  (KL96) edge node {$a,b$} (KL107)
  (KL97) edge node {$a,b$} (KL107)
  (KL93) edge node {$a$} (KL104)
  (KL94) edge node {$a$} (KL105)
  (KL95) edge node {$a$} (KL106)
  (KL103) edge node {$a$} (KL104)
  (KL104) edge node {$a$} (KL105)
  (KL105) edge node {$a$} (KL106)
  (KL106) edge node {$a,b$} (KL107)
  (KL107) edge [loop right] node {$a,b$} (KL107)
  (KL103) edge [loop below] node {$b$} (KL103)
  (KL104) edge [loop below] node {$b$} (KL104)
  (KL105) edge [loop below] node {$b$} (KL105)
  (KL36) edge node {$b$} (KL47)
  (KL56) edge node {$b$} (KL67)
  (KL76) edge node {$b$} (KL87)
  
  ;
  
\end{tikzpicture}
\caption{The automata $A_{10}$, $B_7$ and the reachable part of their cross product $A_{10}\times B_7$.}
\end{figure}
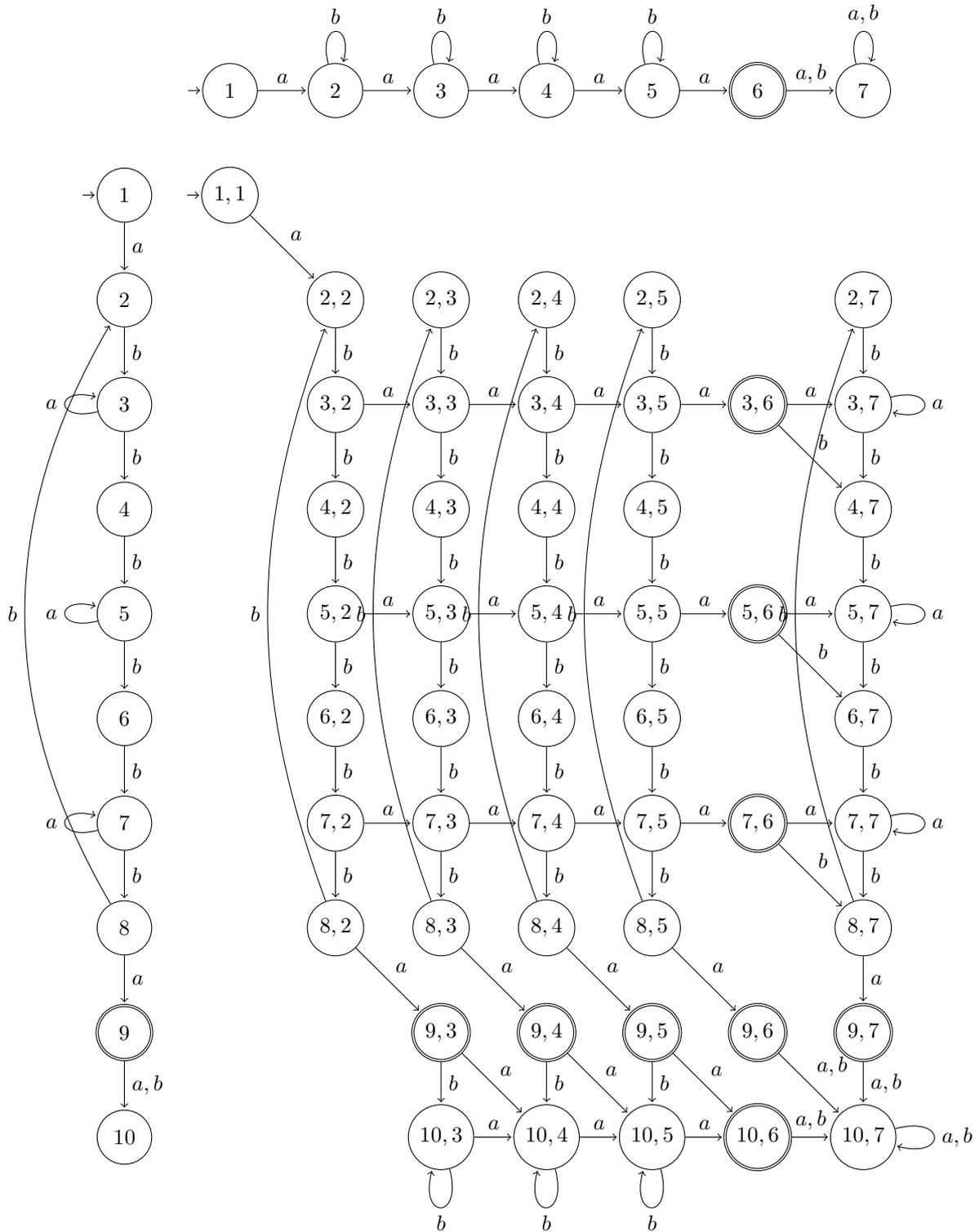

\begin{theorem}
The state complexity of $K_n\cup L_m$ is $mn-(m+n)-2-\lfloor\frac{n-2}{2}\rfloor$.
\end{theorem}
\begin{proof}
Let us consider the reachable states of the cross product automaton $A_n\times B_m$. (See Figure~\ref{fig-cross} for an example.)
It is clear that $(1,1)$ is mapped to $(2,2)$ by $a$ and to $(n,m)$ by $b$.
Also, $(n,m)$ is a trap in $A_n\times B_m$.
Since $1w\neq 1$ for any nonempty $w$ neither in $A_n$ nor in $B_m$, we get that $(1,1)$ is the only reachable state of the form
$(1,j)$ or $(i,1)$ (i.e. the only reachable element of row $1$ and column $1$).
Since $q\cdot_{B_m}b=q$ for $2\leq q\leq m-2$,
and $b$ induces a cycle on $\{2,\ldots,n-2\}$ in $A_n$, we get that any state of the form $(p,q)$ with $2\leq p\leq n-2$ and
$2\leq q\leq m-2$ is reachable if so is $(p',q)$ for all $2\leq p'\leq n-2$.

Thus, reachability of $(2,2)$ implies reachability of $(p,2)$ for each $2\leq p\leq n-2$.

Also, by assumption $n\geq 5$, thus $3\leq n-2$ is an odd number, thus by definition of $A_n$, $3a=3$.
Moreover, $(3,2)$ is reachable. Thus so is $(3,q)$ for each $2\leq q\leq m-2$ (by $q\cdot_{B_m}a=q+1$ for $2\leq q<m-2$).

Hence, so is $(p,q)$ for each $2\leq p\leq n-2$ and $2\leq q\leq m-2$ by the cycle $\{2,\ldots,n-2\}$ induced by $b$ in $A_n$.

Also, $(p,m-2)a=(p,m-1)$ for each odd $3\leq p<n-2$, thus such states $(p,m-1)$ are also reachable.
In particular, $(3,m-1)$ is reachable, hence by $3\cdot_{A_n}a=3$ and $(m-1)a=m$, so is $(3,m)$.
Again by the cycle $\{2,\ldots,n-2\}$ induced by $b$ in $A_n$ we get that every state of the form $(p,m)$ with $2\leq p\leq n-2$
is also reachable.

Now for the last two rows: by $(n-2,p)a=(n-1,p+1)$ for $2\leq p\leq m-2$ we get that $(n-1,p)$ is reachable for $3\leq p\leq m-1$.
Also, $(n-2,m)a=(n-1,m)$ is reachable. By $(n-1,p)b=(n,p)$ for $3\leq p\leq m-2$ we get each $(n,p)$ with $3\leq p\leq m-2$ is
reachable, and so are $(n,m-2)a=(n,m-1)$ and the trap state $(n,m-1)a=(n,m)$ as well.

Overall, out of the $nm$ states the following ones are not reachable:
\begin{itemize}
\item members of the first row or the first column, but $(1,1)$ -- that's $m+n-2$ states;
\item $(n-1,2)$ and $(n,2)$ -- that's two another states;
\item $(p,m-1)$ for even $2\leq p\leq n-2$ -- that's $\lfloor\frac{n-2}{2}\rfloor$ states.
\end{itemize}
So far we have $nm-(n+m)-\lfloor\frac{n-2}{2}\rfloor$ reachable states, some of which might be indistinguishable.
Note that for $K_n\cup L_m$, the accepting states are those of the form $(n-1,q)$ with $3\leq q\leq m$ and $(p,m-1)$ for
odd $3\leq p\leq n-2$ and $p=n-1,n$. 

Amongst the reachable states, $(n,m-1)$, $(n-1,m-1)$ and $(n-1,m)$ are equivalent (accepting only the empty word),
with their merging making an additional reduction of $2$ in the state complexity.

We claim that all the other states are pairwise distinguishable. It is clear that from any state different from $(n,m)$,
either $(n-1,q)$ or $(p,m-1)$ is reachable for some $p\in[n]$ or $q\in[m]$, thus $(n,m)$ is the only empty state.

First we show that final states are pairwise distinguishable (apart from the three one already marked for merging),
by a case analysis. Let $(p,q)\neq(p',q')$ be final states, at least one of them being not a member of $\{(n-1,m),(n-1,m-1),(n,m-1)\}$.
\begin{itemize}
\item {\bf If $q=q'=m-1$}, then without loss of generality we can assume $p<p'$. Then $p<n-1$ (otherwise $(p,q)=(n-1,m-1)$ and $(p',q')=(n,m-1)$
  which are already merged). Hence $ab^{n-2-p}a$ is accepted from $(p,q)$ but not from $(p',q')$.
\item {\bf If $p=p'=n-1$}, then we may assume $q<q'$. Then, $a^{n-1-q}$ is accepted from $(p,q)$ but not from $(p',q')$.
 (Note that $q'a^{n-1-q}\neq n-1$ since $n-1-q>1$ and $K_n$ does not contain any word ending with at least two $a$'s.)
\item {\bf If $p=n-1$, $p'\neq n-1$, $q'=m-1$ and $q<m-1$}, then $(p,q)$ accepts $a^{m-1-q}$. From state $(p',q')$,
  the word leads to a final state only if $p'=n-2$ and $q=m-2$. These states $(n-1,m-2)$ and $(n-2,m-1)$ are distinguishable by
  $ba$, mapping $(n-1,m-2)$ to the final state $(n,m-1)$ and $(n-2,m-1)$ to the empty state.
\item {\bf If $p=n-1$, $p'\neq n-1$, $q'=m-1$ and $q=m$}, then these states $(n-1,m)$ and $(p',m-1)$ with  $p'\neq n$ are distinguishable
  since $(n-1,m)$ accepts only $\varepsilon$, while $(p',m-1)$ accepts $b^{n-2-p'}a$.
\end{itemize}

Next we show that non-final states are also pairwise distinguishable.

Observe that $(1,1)$ is the only non-final nonempty state $(p,q)$ for which $L_{(p,q)}\subseteq a\{a,b\}^*$  (i.e. all the other states
accept some word beginning with $b$). Thus, $(1,1)$ is distinguishable from any other state.
We show that if $(p,q)\neq (p',q')$ are non-final states, then they are distinguishable,
by another lengthy case analysis. Assume $p\leq p'$ and if $p=p'$, then $q<q'$ (that is, $(p,q)$ lexicographically precedes $(p',q')$).
\begin{itemize}
\item {\bf If $p<p'<n-1$ and $\{m-2\}\neq\{q,q'\}$}, then $b^{n-2-p}a$ is accepted only from $p$ and $b^{n-2-p'}a$ is accepted only from $p'$ in $A_n$,
  and by $\{m-2\}\neq\{q,q'\}$, either $q$ or $q'$ does not accept any word of the form $b^*a$, thus 
  $(p,q)$ and $(p',q')$ are distinguishable.
\item {\bf If $p<p'<n-1$ and $q=q'=m-2$, and $p+n-p'$ is even}, then $b^{n-2-p'}a$ takes $(p',q')$ to $(n-1,m-1)$ and $(p,q)$ to
  the state $(p+n-2-p',m-1)$ which we already know to be distinguishable from $(n-1,m-1)$, hence so are $(p,q)$ and $(p',q')$.
\item {\bf If $p<p'<n-1$ and $q=q'=m-2$, and $p+n-p'$ is odd}, then $b^{n-p}a$ takes $(p,q)$ to $(3,m-1)$ and $(p',q')$ to $(n,m-1)$
  which are distinguishable, thus so are $(p,q)$ and $(p',q')$.
\item {\bf If $p<n-1$ and $p'=n$}, then $(p',q')$ accepts all words of the form $b^*a^{m-1-q'}$, in $(n,m-1)$.
  By $pb^{n-2-p}a=3<n-1$ we get that $(p,q)$ cannot accept $b^{n-2-p}a$ in a state merged with $(n,m-1)$ and since we already
  know that final states are distinguishable, so are $(p,q)$ and $(p',q')$.
\item {\bf If $p=p'$ and $q<q'$}, then $a^{m-1-q}$ is accepted from $(p,q)$ but not from $(p',q')$
  unless $p'=n-2$ and $q=m-2$. But also in this case, $(n-2,m-2)$ and $(n-2,m)$ are distinguishable since so are $(n-2,m-2)b=(3,m-2)$ and
  $(n-2,m)b=(3,m)$.
\end{itemize}
Thus, after merging the states $(n-1,m)$, $(n-1,m-1)$ and $(n,m-1)$ we get a reduced automaton having $mn-(m+n)-2-\lfloor\frac{n-2}{2}\rfloor$
states, showing the claim (since $n\leq m$ can be assumed by symmetry).
\end{proof}
We note that for symmetric difference the same construction works, with a minor change: in that case, $(n-1,m-1)$ is not a final state and is
merged by the empty state.
\section*{References}
{
\bibliographystyle{plain}
\bibliography{biblio-factorfree}
}
\end{document}